\documentclass[journal]{IEEEtran}
\usepackage{caption2}
\usepackage{epsfig}

\usepackage[cmex10]{amsmath}
\usepackage{cite}
\usepackage{fixltx2e}
\usepackage{stfloats}
\usepackage{url}
\usepackage{graphicx}
\usepackage{array}
\usepackage{amsfonts}
\usepackage{hyperref}

\def \bI {\mathbb I}
\def \bN {\mathbb N}
\def \bZ {\mathbb Z}
\def \bE {\mathbb E}
\def \bF {\mathbb F}
\def \bR {\mathbb R}

\def \bS {\mathbb S}
\def \bF {\mathbb F}

\def \cB {{\cal B}}

\def \cS {{\cal S}}

\def \cE {{\cal E}}

\def \and {\, \mbox{\rm and}\, }
\def \sinc {\,{\rm sinc}\,}

\def \supp {\,{\rm supp}\,}

\newtheorem{theorem}{\it Theorem}
\newtheorem{lemma}[theorem]{\it Lemma}

\newenvironment{proof}{\noindent{\it Proof:}}{\quad \hfill$\Box$\vspace{2ex}}

\begin{document}

\title{Convergence Analysis of the Gaussian Regularized Shannon Sampling Formula}


\author{Rongrong~Lin~
        and~Haizhang~Zhang 
\thanks{This work was supported in part by Natural Science Foundation of China under grants 11222103 and 11101438.}
\thanks{R. Lin is with the School of Mathematics and Computational Science, Sun Yat-sen University, Guangzhou 510275, P. R. China (e-mail: linrr@mail2.sysu.edu.cn).}
\thanks{H. Zhang (corresponding author) is with the School of Mathematics and Computational
Science and Guangdong Province Key Laboratory of Computational Science, Sun Yat-sen University, Guangzhou 510275, P. R. China (e-mail: zhhaizh2@sysu.edu.cn).}}
\maketitle

\begin{abstract}
We consider the reconstruction of a bandlimited function from its finite localized sample data. Truncating the classical Shannon sampling series results in an unsatisfactory convergence rate due to the slow decayness of the sinc function. To overcome this drawback, a simple and highly effective method, called the Gaussian regularization of the Shannon series, was proposed in engineering and has received remarkable attention.
It works by multiplying the sinc function in the Shannon series with a regularization Gaussian function. L. Qian (Proc. Amer. Math. Soc., 2003) established the convergence rate of $O(\sqrt{n}\exp(-\frac{\pi-\delta}2n))$ for this method, where $\delta<\pi$ is the bandwidth and $n$ is the number of sample data. C. Micchelli {\it et al.} (J. Complexity, 2009) proposed a different regularization method and obtained the corresponding convergence rate of $O(\frac1{\sqrt{n}}\exp(-\frac{\pi-\delta}2n))$. This latter rate is by far the best among all regularization methods for the Shannon series. However, their regularized function involves the solving of a linear system and is implicit and more complicated. The main objective of this note is to show that the Gaussian regularized Shannon series can also achieve the same best convergence rate as that by C. Micchelli {\it et al}. We also show that the Gaussian regularization method can improve the convergence rate for the useful average sampling. Numerical experiments are presented to justify the obtained results.
\end{abstract}


%

\section{Introduction}
%
%
%
%
\IEEEPARstart{T}{he} main purpose of this paper is to show that the Gaussian regularized Shannon sampling formula to reconstruct a bandlimited function can achieve by far the best convergence rate among all regularization methods for the Shannon sampling series in the literature.
As a result, we improve L. Qian's error estimate \cite{Qian} for this highly successful method in engineering.

We first introduce the {\it Paley-Wiener space} $\cB_\delta(\bR^d)$ with the bandwidth $\delta:=(\delta_1,\delta_2,\dots,\delta_d)\in (0,+\infty)^d$ defined as
$$
\cB_\delta(\bR^d):=\{f\in L^2(\bR^d)\cap C(\bR^d):\supp\hat{f}\subseteq[-\delta,\delta]\},
$$
where $[-\delta,\delta]:=\prod_{k=1}^d[-\delta_k,\delta_k]$. In this paper, the Fourier transform of $f\in L^1(\bR^d)$ takes the form
$$
\hat{f}(\xi):=\frac{1}{(\sqrt{2\pi})^d}\int_{\bR^d} f(t)e^{-i t\cdot\xi}dt,\ \xi\in\bR^d,
$$
where $t\cdot\xi$ denotes the standard inner product of $t$ and $\xi$ in $\bR^d$. The classical {\it Shannon sampling theorem} \cite{Shannon, Whittaker} states that each $f\in \cB_\pi(\bR)$ can be completely reconstructed from its infinite sampling data $\{f(j):j\in\bZ\}$. Specifically, it holds
\begin{equation}\label{Shannonseries}
f(t)=\sum_{j\in\bZ}f(j)\sinc(t-j), \ t\in\bR, \ f\in\cB_\pi(\bR),
\end{equation}
where $\sinc(x)=\sin(\pi x)/(\pi x)$. Many generalizations of Shannon's sampling theorem have been established (see, for example, \cite{Aldroubi,Grochenig,Jerri,Hoffman,Qianthesis, Sun1,Sun2,Unser,Zhang09} and the references therein).

In practice, we can only sum over finite sample data ``near'' $t$ in (\ref{Shannonseries}) to approximate $f(t)$. Truncating the series (\ref{Shannonseries}) results in a convergence rate of the order $O(\frac{1}{\sqrt{n}})$ due to the slow decayness of the sinc function, \cite{Helms,Jagerman,Jordon,Tsybakov}. Besides, this truncated series was proved to be the optimal reconstruction method in the worst case scenario in $\cB_\pi(\bR)$.
Dramatic improvement of the convergence rate can be achieved when $f\in \cB_\delta(\bR)$, $\delta<\pi$.
In this case, $\{f(j):j\in\bZ\}$ turns out to be a set of oversampling data, where oversampling means to sample at a rate strictly larger than the Nyquist rate $\frac{\delta}{\pi}<1$.

Three explicit methods \cite{Qian,Jagerman,Micchelli} have been proposed in order to reconstruct a univariate bandlimited function $f\in\cB_\delta(\bR)$ ($\delta<\pi$) from its finite oversampling data $\{f(j):-n+1\le j\le n\}$ with an exponentially decaying approximation error. They work by multiplying the sinc function with a rapidly-decaying regularization function. Jagerman \cite{Jagerman} used a power of the sinc functions as the regularizer and obtained the convergence rate of
$$
O\Big(\frac1n\exp(-\frac{\pi-\delta}en)\Big).
$$
Micchelli {\it et al.} \cite{Micchelli} chose a spline function as the regularizer and attained the convergence rate
\begin{equation}\label{bestrate}
O\Big(\frac1{\sqrt{n}}\exp(-\frac{\pi-\delta}2n)\Big),
\end{equation}
which is by far the best convergence rate among all regularization methods for the Shannon sampling series in reconstructing a bandlimited function. However, the spline regularization function in \cite{Micchelli} is implicit and involves the solving of a linear system. A third method using a Gaussian function as the regularizer was first proposed by Wei \cite{Wei98}. Precisely, the {\it Gaussian regularized Shannon sampling series} in \cite{Wei98} is defined as
\begin{equation}\label{GRSseries}
(S_{n,r}f)(t):=\sum_{j=-n+1}^n f(j)\sinc(t-j) e^{-\frac{(t-j)^2}{2r^2}},t\in(0,1),f\in\cB_\delta(\bR).
\end{equation}
A convergence analysis of the method was presented by Qian \cite{Qian,Qianthesis}. Taking a small computational mistake in (2.33) in \cite{Qian} into consideration and optimizing about the variance $r$ of the Gaussian function as done in \cite{Micchelli}, Qian actually established the following convergence rate for (\ref{GRSseries})
\begin{equation}\label{qianrate}
O\Big(\sqrt{n}\exp(-\frac{\pi-\delta}2n)\Big).
\end{equation}

Due to its simplicity and high accuracy (\ref{qianrate}), the Gaussian regularized Shannon sampling series (\ref{GRSseries}) has been widely applied to scientific and engineering computations. In fact, more than a hundred such papers have appeared (see \url{http://www.math.msu.edu/~wei/pub-sec.html} for the list, and \cite{Wei98,Wei00,Wei000,Wei97,Wei08,Zhao} for comments and discussion).
We note that although the exponential term in (\ref{qianrate}) is as good as that in (\ref{bestrate}) for the spline function regularization \cite{Micchelli}, the first term $\sqrt{n}$ is much worse than $\frac1{\sqrt{n}}$ in (\ref{bestrate}). The first purpose of this paper is to show that the convergence rate of the Gaussian regularized Shannon sampling series can be improved to (\ref{bestrate}).
Thus, this method enjoys both simplicity and the best convergence rate by far. This will be done in Section II, where we also improve the convergence rates of the Gaussian regularized Shannon sampling series for derivatives and multivariate bandlimited functions.
In Section III, we show that the Gaussian regularization method can also improve the convergence rate for the useful average sampling. In the last section, we demonstrate our results via several numerical experiments.

\section{Improved Convergence Rates for the Gaussian Regularization}

In this section, we improve the convergence rate analysis for the Gaussian regularized Shannon sampling series. Specifically, we shall show that it can achieve by far the best rate (\ref{bestrate}) for univariate bandlimited functions and its derivatives, and for multivariate bandlimited functions. We separate the three cases into different subsections.

\subsection{Univariate Bandlimited Functions}

Let $f\in\cB_\delta(\bR)$ with $\delta\in(0,\pi)$ throughout this subsection. We shall use the Gaussian regularized Shannon sampling series (\ref{GRSseries}) to reconstruct the values of $f$ at $t\in(0,1)$ from the finite localized oversampling data $\{f(j):-n+1\le j\le n\}$. We shall need a few technical facts in order to improve the convergence rate established in \cite{Qian}. They were also frequently used in the estimates in \cite{Qian,Qianthesis}.

Firstly, it is well-known that $\sinc(\cdot-j)$, $j\in\bZ$ form an orthonormal basis for $\cB_\pi(\bR)$. As $f\in\cB_\delta(\bR)\subseteq\cB_\pi(\bR)$, we have by the Parseval identity
\begin{equation}\label{Parseval}
\|f\|^2_{L^2(\bR)}=\sum_{j\in\bZ}|f(j)|^2,\ \ f\in\cB_\delta(\bR).
\end{equation}
The second result needed is the following upper bound estimate of Mills' ratio \cite{Pollak} of standard normal law:
\begin{equation}\label{lemma1}
\int_x^{+\infty} e^{-t^2}dt< \frac{e^{-x^2}}{2x},\ \ x>0.
\end{equation}
We shall also need a variant of its discrete version:
\begin{equation}\label{lemma2}
\sum_{j\notin(-n,n]}e^{-\frac{(t-j)^2}{r^2}}<\frac{r^2}{n-1}e^{-\frac{(n-1)^2}{r^2}}, r>0, n\ge2, t\in(0,1).
\end{equation}
The next two are a useful computation of the Fourier transform
\begin{equation}\label{lemma3}
\Big(\sinc(t-j) e^{-\frac{(t-j)^2}{2r^2}}\Big)\hat{\,}(\xi)=\frac{e^{-ij\xi}}{2\pi}\int_{\xi-\pi}^{\xi+\pi} re^{-\frac{r^2\eta^2}{2}}d\eta,\ \xi\in\bR
\end{equation}
and an associated estimate
\begin{equation}\label{lemma4}
\Big|1-\frac{1}{\sqrt{\pi}}\int_{\frac{(\xi-\pi)r}{\sqrt{2}}}^{\frac{(\xi+\pi)r}{\sqrt{2}}} e^{-\tau^2}d\tau\Big|\le  \sqrt{\frac{2}{\pi}}\frac{e^{-\frac{(\pi-\delta)^2r^2}{2}}}{(\pi-\delta)r}\mbox{ for all }\xi\in[-\delta,\delta].
\end{equation}
The last one is the upper bound estimate
\begin{equation}\label{lemma5}
|H_k(x)|\le (2x)^k,\ \ |x|\ge \frac{k}{2}
\end{equation}
for the $k$-th order Hermite polynomial defined as
$$
H_k(x):=k!\sum_{i=0}^{\lfloor \frac{k}{2}\rfloor}(-1)^i \frac{(2x)^{k-2i}}{i!(k-2i)!},\ \ x\in\bR.
$$

We are in a position to present an improved convergence rate analysis for the Gaussian regularized Shannon sampling series. We follow the methods in \cite{Qian} and emphasize that the improvement is achieved by applying the Cauchy-Schwartz inequality to obtain a better estimate for (2.5) in \cite{Qian}.

\begin{theorem}\label{Theorem1} Let $\delta\in (0,\pi)$ , $n\ge 2$, and choose $r:=\sqrt{\frac{n-1}{\pi-\delta}}$. The Gaussian regularized Shannon sampling series (\ref{GRSseries}) satisfies
\begin{equation}\label{mainrate}
\begin{array}{ll}
&\displaystyle{\sup_{f\in\cB_\delta(\bR),\ \|f\|_{L^2(\bR)}\le 1} \|f-S_{n,r}f\|_{L^\infty((0,1))}}\\
&\displaystyle{\le \left( \sqrt{2\delta}+\frac{1}{\sqrt{n}}\right) \frac{ e^{-\frac{(\pi-\delta)(n-1)}{2}}}{\pi\sqrt{(\pi-\delta)(n-1)}}.}
\end{array}
\end{equation}
\end{theorem}
\begin{proof}  Let $f\in\cB_\delta(\bR)$ with $\|f\|_{L^2(\bR)}\le 1$. Set $(f-S_{n,r}f)(t):=E_1(t)+E_2(t)$, $t\in(0,1)$, where
$$
\begin{array}{ll}
&\displaystyle{E_1(t):=f(t)-\sum_{j\in\bZ} f(j)\sinc(t-j)e^{-\frac{(t-j)^2}{2r^2}},}\\
&\displaystyle{E_2(t):=\sum_{j\notin(-n,n]} f(j)\sinc(t-j) e^{-\frac{(t-j)^2}{2r^2}}.}
\end{array}
$$
Bound $E_1(t)$ by its Fourier transform as follows
$$
|E_1(t)|\le \frac1{\sqrt{2\pi}}\|\hat{E_1}\|_{L^1(\bR)},\ \ t\in(0,1).
$$
Computing and bounding $\hat{E_1}$ by (\ref{lemma3}), (\ref{lemma4}) and similar arguments as those in \cite{Qian}, we have
\begin{equation}\label{E11}
|E_1(t)|\le  \frac{\sqrt{2\delta}e^{-\frac{(\pi-\delta)^2r^2}{2}}}{\pi (\pi-\delta)r},\ t\in(0,1).
\end{equation}
Observe that
$$
|\sinc(t-j)|\le \frac1{\pi n},\ t\in(0,1),\ j\notin(-n,n].
$$
Thus,
$$
\begin{array}{ll}
|E_2(t)|
&\displaystyle{\le \sum_{j\notin(-n,n]} |f(j)|\left|\frac{\sin\pi(t-j)}{\pi(t-j)}\right| e^{-\frac{(t-j)^2}{2r^2}}}\\
&\displaystyle{\le \frac{1}{\pi n}\sum_{j\notin(-n,n]} |f(j)| e^{-\frac{(t-j)^2}{2r^2}}, \ t\in(0,1).}
\end{array}
$$
Apply the Cauchy-Schwartz inequality, we get by (\ref{Parseval}) and (\ref{lemma2})
\begin{equation}\label{E12}
\begin{array}{ll}
|E_2(t)|
&\displaystyle{\le  \frac{1}{\pi n}\Big(\sum_{j\notin(-n,n]} |f(j)|^2\Big)^{\frac12}\Big(\sum_{j\notin(-n,n]}e^{-\frac{(t-j)^2}{r^2}} \Big)^{\frac12}}\\
&\displaystyle{\le \frac{re^{-\frac{(n-1)^2}{2r^2}}}{\pi n\sqrt{n-1}}, \ t\in(0,1).}
\end{array}
\end{equation}
Combining (\ref{E11}) with (\ref{E12}) and optimally choosing $r=\sqrt{\frac{n-1}{\pi-\delta}}$ completes the proof.
\end{proof}

We remark that the estimate (\ref{mainrate}) is of the same order as the best convergence rate (\ref{bestrate}) by far in the literature. A second remark is on the degenerated case when $\delta=\pi$. In this case, the estimate in \cite{Qian} or the above (\ref{mainrate}) is apparently meaningless. To make up for the drawback, a more delicate upper bound estimate is needed for $E_1$. Specifically, we have
$$
\begin{array}{ll}
|E_1(t)|
&\displaystyle{\le \frac{1}{\sqrt{2\pi}}\Big(\int_{[-\pi,\pi]\setminus[-\pi+n^{-\frac{3}{4}},\pi-n^{-\frac{3}{4}}]} |\hat{E}_1(\xi)|d\xi\Big)}\\
&\displaystyle{\quad + \frac{1}{\sqrt{2\pi}}\int_{-\pi+n^{-\frac{3}{4}}}^{\pi-n^{-\frac{3}{4}}} |\hat{E}_1(\xi)|d\xi}\\
&\displaystyle{\le \frac{1}{\sqrt{\pi}}\frac{1}{n^{\frac38}}+  \sqrt{\frac{2}{\pi}} \frac{n^{\frac34}}{r}.}
\end{array}
$$
Taking $r=n^{\frac{9}{8}}$ above and using (\ref{E12}), we obtain the convergence rate $3n^{-\frac38}$ for $f\in\cB_\pi(\bR)$.

\subsection{Derivatives of Univariate Bandlimited Functions}

By the Paley-Wiener theorem, each function in $\cB_\delta(\bR)$ is infinitely differentiable. In this subsection, we are concerned with the reconstruction of the derivatives of $f\in \cB_\delta(\bR)$ by the Gaussian regularized Shannon sampling series (\ref{GRSseries}). The convergence rate obtained in \cite{Qian,Qianthesis} is also of the order (\ref{qianrate}). We shall improve the estimate.
\begin{theorem}\label{Theorem2}
Let $s\in\bN$, $n\ge \max\{3,\frac{s^2}{2(\pi-\delta)}\}$, and $r:=\sqrt{\frac{n-2}{\pi-\delta}}$. It holds
$$
\begin{array}{ll}
&\displaystyle{\sup_{f\in\cB_\delta(\bR),\|f\|_{L^2(\bR)}\le 1}\|f^{(s)}-(S_{n,r}f)^{(s)}\|_{L^\infty((0,1))}}\\
&\displaystyle{\le
\left( \frac{\sqrt{2}\delta^{s+\frac12}}{\sqrt{2s+1}} +\frac{24(s+2)!}{\sqrt{n}}   \right)\frac{e^{-\frac{(\pi-\delta)(n-2)}{2}}}{\pi\sqrt{(\pi-\delta)(n-2)}}.}
\end{array}
$$
\end{theorem}
\begin{proof} We may suppose that $f\in\cB_\delta(\bR)$ with $\|f\|_{L^2(\bR)}\le 1$. Set $f^{(s)}(t)-(S_{n,r}f)^{(s)}(t):=\bE_1(t)+\bE_2(t)$, $t\in(0,1)$, where
$$
\begin{array}{ll}
&\displaystyle{\bE_1(t):=f^{(s)}(t)-\frac{1}{\sqrt{2\pi}}\sum_{j\in\bZ}f(j)\Big(\sinc(t-j)e^{-\frac{(t-j)^2}{2r^2}}\Big)^{(s)},}\\
&\displaystyle{\bE_2(t):=\frac{1}{\sqrt{2\pi}}\sum_{j\notin(-n,n]}f(j)\Big(\sinc(t-j)e^{-\frac{(t-j)^2}{2r^2}}\Big)^{(s)}.}
\end{array}
$$
It has been estimated in \cite{Qian} that
\begin{equation}\label{bE1}
\|\bE_1\|_{L^\infty((0,1))} \le  \sqrt{\frac{2}{2s+1}}\frac{\delta^{s+\frac12} e^{-\frac{(\pi-\delta)^2r^2}{2}}}{\pi(\pi-\delta)r}.
\end{equation}
For each $t\in(0,1)$, we calculate
$$
\begin{array}{ll}
\displaystyle{\bE_2(t)}
&\displaystyle{=\frac{1}{\sqrt{2\pi}}\sum_{j\notin(-n,n]}f(j)}\\
&\displaystyle{\quad\cdot\left[\sum_{k=0}^s \frac{s!}{k!(s-k)!} \Big(\frac{\sin (\pi t-j\pi)}{\pi(t-j)}\Big)^{(s-k)}\Big(e^{-\frac{(t-j)^2}{2r^2}}\Big)^{(k)} \right]}\\
&\displaystyle{=\frac{1}{\sqrt{2\pi}}\sum_{j\notin(-n,n]}f(j)\left[\sum_{k=0}^s \frac{s!}{k!}\right.}\\
&\displaystyle{\quad\cdot \Biggl(\sum_{l=0}^{s-k}\frac{\pi^{s-k-l-1}\sin\Big(\pi(t-j+\frac{(s-k-l)}{2})\Big)}{l!(s-k-l)!}\cdot\frac{(-1)^l l!}{(t-j)^{l+1}}\Biggr)}\\
&\displaystyle{\quad\cdot\left.\frac{(-1)^k H_k(\frac{t-j}{\sqrt{2}r})}{(\sqrt{2}r)^k}\right] e^{-\frac{(t-j)^2}{2r^2}}.}\\
\end{array}
$$
Noticing the simple fact
\begin{equation}\label{Simplefact}
\frac{1}{|t-j|^{l+1}} \le \frac1n \mbox{ for }j\notin(-n,n],\ t\in(0,1),\ l\in\bZ_+,
\end{equation}
we obtain by (\ref{lemma5})
$$
\begin{array}{ll}
\displaystyle{|\bE_2(t)|}
&\displaystyle{\le\frac{s!}{\pi\sqrt{2\pi}}\sum_{j\notin(-n,n]}|f(j)| }\\
&\displaystyle{\quad\cdot\left[\sum_{k=0}^s\frac{1}{k!}\Big(\sum_{l=0}^{s-k}\frac{\pi^{s-k-l}}{(s-k-l)!}\frac{1}{|t-j|^{l+1}}\Big)\frac{|t-j|^k}{r^{2k}}\right] e^{-\frac{(t-j)^2}{2r^2}}}\\
&\displaystyle{\le\frac{s!}{n\pi\sqrt{2\pi}}\sum_{|j|\ge n}|f(j+1)|}\\
&\displaystyle{\quad\cdot\left[\sum_{k=0}^s\frac{|j|^k}{k!r^{2k}}\cdot\Big(\sum_{l=0}^{s-k}\frac{\pi^{s-k-l}}{(s-k-l)!}\Big)\right]e^{-\frac{j^2}{2r^2}}}.
\end{array}
$$
Since $e^x>\frac{x^j}{j!}$ for all $x>0$ and $j\in\bZ_+$, we have
$$
\begin{array}{ll}
|\bE_2(t)|
&\displaystyle{\le\frac{s!}{n\pi\sqrt{2\pi}}\sum_{|j|\ge n} |f(j+1)|\Big(\sum_{k=0}^s e^{\frac{|j|}{r^2}}(s-k+1)e^\pi\Big)e^{-\frac{j^2}{2r^2}}}\\
&\displaystyle{\le\frac{e^\pi (s+2)!}{n(2\pi)^\frac32}\sum_{|j|\ge n} |f(j+1)|e^{\frac{-j^2+2|j|}{2r^2}}.}\\
\end{array}
$$
Note that $n\ge 1+\frac{s^2}{2(\pi-\delta)}$ and $r=\sqrt{\frac{n-2}{\pi-\delta}}$ implies $n\ge \frac{sr}{\sqrt{2}}$. We thus get by (\ref{Parseval}), (\ref{lemma2}) and the Cauchy-Schwartz inequality
\begin{equation}\label{bE2}
\begin{array}{ll}
\|\bE_2(t)\|_{L^\infty((0,1))}
&\displaystyle{\le\frac{e^\pi (s+2)!e^{\frac{1}{2r^2}}}{n(2\pi)^\frac32}\Big(2\sum_{j\ge n-1} e^{-\frac{j^2}{r^2}}\Big)^{\frac12}}\\
&\displaystyle{\le \frac{e^\pi (s+2)!e^{\frac{1}{2r^2}}}{(2\pi)^\frac32}\frac{re^{-\frac{(n-2)^2}{2r^2}}}{n\sqrt{n-2}}.}
\end{array}
\end{equation}
Substituting $r=\sqrt{\frac{n-2}{\pi-\delta}}$ into (\ref{bE1}) and (\ref{bE2}), we have
$$
\begin{array}{ll}
&\displaystyle{\|\bE_2\|_{L^\infty((0,1))}+\|\bE_2\|_{L^\infty((0,1))}}\\
&\displaystyle{=\Big( \frac{\sqrt{2}}{\pi\sqrt{2s+1}} \frac{\delta^{s+\frac12}}{\sqrt{(\pi-\delta)(n-2)}} +\frac{e^\pi (s+2)!e^{\frac{\pi-\delta}{2(n-2)}}}{(2\pi)^{\frac32}\sqrt{\pi-\delta}n}  \Big)}\\
&\displaystyle{\quad\cdot e^{-\frac{(\pi-\delta)(n-2)}{2}} }\\
&\displaystyle{<\Big( \frac{\sqrt{2}\delta^{s+\frac12}}{\sqrt{2s+1}} +\frac{24(s+2)!}{\sqrt{n}}   \Big)\frac{e^{-\frac{(\pi-\delta)(n-2)}{2}}}{\pi\sqrt{(\pi-\delta)(n-2)}} }.
\end{array}
$$
In the last inequality, we use $r=\sqrt{\frac{n-2}{\pi-\delta}}\ge \frac{1}{\sqrt{\pi}}$ and $\frac{e^{\frac{3\pi}{2}}}{2\sqrt{2\pi}}\le 24$. The proof is complete.
\end{proof}

\subsection{Multivariate Bandlimited Functions}

In this subsection, we show that the best convergence rate (\ref{bestrate}) can also be achieved in the reconstruction of a multivariate bandlimited function by the Gaussian regularized Shannon sampling series. Techniques to be used for this case is quite different from those in the univariate case and also much differs from those in \cite{Qianthesis}.

Let $\delta:=(\delta_1,\delta_2,\dots,\delta_d)\in(0,\pi)^d$ and $J_n:=\{j\in\bZ^d: j\in(-n,n]^d\}$ in this subsection. The Gaussian regularized Shannon sampling series \cite{Qianthesis} to reconstruct a multivariate function $f\in \cB_\delta(\bR^d)$ from its finite sample data $f(J_n)$ is defined as
\begin{equation}\label{multiLocalsampling}
(\cS_{n,r}f)(t):=\sum_{j\in J_n} f(j)\sinc(t-j)e^{-\frac{\|t-j\|^2}{2r^2}},t\in(0,1)^d,
\end{equation}
where
$$
\sinc(x):=\prod_{k=1}^d \frac{\sin(\pi x_k)}{\pi x_k}, \ \ x\in\bR^d
$$
and $\|x\|$ denotes the standard Euclidean norm of $x$ in $\bR^d$.

Two lemmas are needed to present an improved convergence analysis for (\ref{multiLocalsampling}).

\begin{lemma}\label{lemma6} Let $r>0$ and $n\ge2$. It holds for all $t\in(0,1)^d$
$$
\sum_{j\notin J_n} \Big(\sinc^2(t-j)e^{-\frac{\|t-j\|^2}{r^2}}\Big)\le \frac{d}{\pi^2} \frac{r^2e^{-\frac{(n-1)^2}{r^2}}}{n^2(n-1)}.
$$
\end{lemma}
\begin{proof}
Observe that
\begin{equation}\label{Subset}
\{j\in\bZ^d: j\notin J_n\}\subseteq \bigcup_{k=1}^d \{j\in\bZ^d: j_k\notin (-n,n]\}.
\end{equation}
By (\ref{Subset}), we have for each $t\in(0,1)^d$
$$
\begin{array}{ll}
&\displaystyle{\sum_{j\notin J_n}\Big(\sinc^2(t-j)e^{-\frac{\|t-j\|^2}{r^2}}\Big)}\\
&\displaystyle{\le \sum_{k=1}^d\Big(\sum_{j_k\notin (-n,n]} \sinc^2(t_k-j_k)e^{-\frac{(t_k-j_k)^2}{r^2}}\Big) }\\
&\displaystyle{\quad\cdot\prod_{l\ne k}\Big(\sum_{j_l\in\bZ}\sinc^2(t_l-j_l)e^{-\frac{(t_l-j_l)^2}{r^2}}\Big).}
\end{array}
$$
Note that $\sum_{m\in\bZ}\sinc^2(x-m)=1$, $x\in\bR$. Thus, we have
$$
\begin{array}{ll}
&\displaystyle{\sum_{j\notin J_n}\prod_{k=1}^d \sinc^2(t_k-j_k)e^{-\frac{(t_k-j_k)^2}{r^2}}}\\
&\displaystyle{\le \sum_{k=1}^d\Big(\sum_{j_k\notin (-n,n]} \sinc^2(t_k-j_k)e^{-\frac{(t_k-j_k)^2}{r^2}}\Big).}\\
\end{array}
$$
Applying (\ref{lemma2}) to the last inequality gives the desired result.
\end{proof}

Introduce the important constant
$$
\Delta:=\max\{\delta_k:k=1,2,\dots,d\}.
$$
\begin{lemma}\label{lemma8} Let $\delta\in(0,\pi)^d$. It holds for $\xi\in\prod_{k=1}^d[-\delta_k,\delta_k]$
$$
1-\prod_{k=1}^d\frac{1}{\sqrt{\pi}}\int_{\frac{(\xi_k-\pi)r}{\sqrt{2}}}^{\frac{(\xi_k+\pi)r}{\sqrt{2}}} e^{-\tau^2}d\tau
\le \sqrt{\frac{2}{\pi}} \frac{de^{-\frac{(\pi-\Delta)^2r^2}{2}}}{(\pi-\Delta)r}.
$$
\end{lemma}
\begin{proof} By (\ref{lemma1}), we have for $\xi\in\prod_{k=1}^d[-\delta_k,\delta_k]$
$$
\begin{array}{ll}
&\displaystyle{\frac{1}{\sqrt{\pi}}\int_{\frac{(\xi_k-\pi)r}{\sqrt{2}}}^{\frac{(\xi_k+\pi)r}{\sqrt{2}}} e^{-\tau^2}d\tau}\\
&\displaystyle{=1-\frac{1}{\sqrt{\pi}}\int_{\frac{(\pi-\xi_k)r}{\sqrt{2}}}^{+\infty} e^{-\tau^2}d\tau-\frac{1}{\sqrt{\pi}}\int_{\frac{(\pi+\xi_k)r}{\sqrt{2}}}^{+\infty} e^{-\tau^2}d\tau}\\
&\displaystyle{\ge 1-\frac{1}{\sqrt{\pi}}\left(\frac{e^{-\frac{(\pi-\xi_k)^2r^2}{2}}}{\sqrt{2}(\pi-\xi_k)r}
+\frac{e^{-\frac{(\pi+\xi_k)^2r^2}{2}}}{\sqrt{2}(\pi+\xi_k)r}\right). }\\
\end{array}
$$
We then apply the elementary fact that for constants $\tau_k>0$, $1\le k\le d$ with $\sum_{k=1}^d \tau_k<1$, it holds
$$
1-\prod_{k=1}^d (1-\tau_k)\le\sum_{k=1}^d \tau_k.
$$
As a consequence,
$$
\begin{array}{ll}
&\displaystyle{1-\prod_{k=1}^d\frac{1}{\sqrt{\pi}}\int_{\frac{(\xi_k-\pi)r}{\sqrt{2}}}^{\frac{(\xi_k+\pi)r}{\sqrt{2}}} e^{-\tau^2}d\tau}\\
&\displaystyle{\le \frac{1}{\sqrt{\pi}}\sum_{k=1}^d\left(\frac{e^{-\frac{(\pi-\xi_k)^2r^2}{2}}}{\sqrt{2}(\pi-\xi_k)r}
+\frac{e^{-\frac{(\pi+\xi_k)^2r^2}{2}}}{\sqrt{2}(\pi+\xi_k)r}\right)}\\
&\displaystyle{\le  \sqrt{\frac{2}{\pi}}\sum_{k=1}^d\frac{e^{-\frac{(\pi-\xi_k)^2r^2}{2}}}{(\pi-\xi_k)r }.}
\end{array}
$$
The proof is completed by noting that $\frac1x e^{-\frac{r^2x^2}{2}}$ is decreasing on $(0,+\infty)$.
\end{proof}

With the above preparation, we have the following main theorem.
\begin{theorem}\label{Theorem3} Let $\delta\in (0,\pi)^d$, $n\ge2$, and $r:=\sqrt{\frac{n-1}{\pi-\Delta}}$. The multivariate Gaussian regularized Shannon sampling series (\ref{multiLocalsampling}) satisfies
$$
\begin{array}{ll}
&\displaystyle{\sup_{f\in\cB_\delta(\bR^d),\ \|f\|_{L^2(\bR^d)}\le 1} \|f-\cS_{n,r}f\|_{L^\infty((0,1)^d)}}\\
&\displaystyle{\le\left(  d(2\Delta)^{\frac{d}{2}}+\sqrt{\frac{d}{n}}\right) \frac{e^{-\frac{(\pi-\Delta)(n-1)}{2}}}{\pi\sqrt{(\pi-\Delta)(n-1)}}.}
\end{array}
$$
\end{theorem}
\begin{proof} Let $f\in\cB_\delta(\bR^d)$ with $\|f\|_{L^2(\bR^d)}\le 1$. Set $f(t)-(S_{n,r}f)(t):=\cE_1(t)+\cE_2(t)$, $t\in(0,1)^d$,
where
$$
\begin{array}{ll}
&\displaystyle{\cE_1(t):=f(t)-\sum_{j\in\bZ^d} f(j)\sinc(t-j) e^{-\frac{\|t-j\|^2}{2r^2}},}\\
&\displaystyle{\cE_2(t):=\sum_{j\notin J_n} f(j) \sinc(t-j) e^{-\frac{\|t-j\|^2}{2r^2}}.}
\end{array}
$$
By Lemma \ref{lemma8} and similar arguments as those in \cite{Qian} for the univariate case, we have
$$
\begin{array}{ll}
|\hat{\cE}_1(\xi)|
&\displaystyle{= |\hat{f}(\xi)|\left|1-\prod_{k=1}^d\frac{1}{\sqrt{\pi}}\int_{\frac{(\xi_k-\pi)r}{\sqrt{2}}}^{\frac{(\xi_k+\pi)r}{\sqrt{2}}} e^{-\tau^2}d\tau\right|}\\
&\displaystyle{\le |\hat{f}(\xi)|\sqrt{\frac{2}{\pi}}\frac{de^{-\frac{(\pi-\Delta)^2r^2}{2}}}{(\pi-\Delta)r},\ \xi\in[-\delta,\delta].}
\end{array}
$$
Bounding $\cE_1$ by its Fourier transform and then applying the Cauchy-Schwartz inequality, we get
\begin{equation}\label{E231}
\begin{array}{ll}
\|\cE_1\|_{L^\infty((0,1)^d)}
&\displaystyle{\le \frac{1}{\sqrt{2\pi}}\int_{[-\delta,\delta]} |\hat{\cE}(\xi)|d\xi}\\
&\displaystyle{\le \frac{d}{\pi}\frac{e^{-\frac{(\pi-\Delta)^2r^2}{2}}}{(\pi-\Delta)r}\int_{[-\delta,\delta]}|\hat{f}(\xi)|d\xi}\\
&\displaystyle{\le \frac{d(2\Delta)^{\frac{d}{2}}}{\pi}\frac{e^{-\frac{(\pi-\Delta)^2r^2}{2}}}{(\pi-\Delta)r}.}
\end{array}
\end{equation}
Recall that
$$
\|f\|_{L^2(\bR^d)}^2=\sum_{j\in\bZ^d}|f(j)|^2,\ \ f\in\cB_\delta(\bR^d).
$$
By this, Lemma \ref{lemma6}, and the Cauchy-Schwartz inequality, we obtain
\begin{equation}\label{E232}
\begin{array}{ll}
\|\cE_2\|_{L^\infty((0,1)^d)}
&\displaystyle{\le \Big(\sum_{j\notin J_n} |f(j)|^2\Big)^{\frac12}}\\
&\displaystyle{\quad\cdot \Big(\sum_{j\notin J_n} \sinc^2(t-j) e^{-\frac{\|t-j\|^2}{r^2}} \Big)^{\frac12}}\\
&\displaystyle{\le \frac{\sqrt{d}}{\pi} \frac{re^{-\frac{(n-1)^2}{2r^2}}}{n\sqrt{n-1}}.}
\end{array}
\end{equation}
The proof is completed by taking $r=\sqrt{\frac{n-1}{\pi-\Delta}}$ in (\ref{E231}) and (\ref{E232}).
\end{proof}

\section{Gaussian Regularization for Average Sampling}

In this section, we will apply the method of Gaussian regularization to the useful average sampling. A main purpose is again to improve the convergence rate. We first introduce some basic facts about the average sampling.

Sampling a function $f$ at $j\in\bZ$ can be viewed as applying the Delta distribution to the function $f(j+\cdot)$. The Delta distribution is used theoretically but hard to implement physically. Hence, a practical way is to approximate the Delta distribution by an averaging function with small support around the origin. For this sake, we consider the following average sampling strategy:
$$
\tilde{f}(j):=\int_{-\sigma/2}^{\sigma/2}f(j+\cdot)d\nu(x),\ j\in\bZ,\ f\in\cB_\delta(\bR),
$$
where $0<\sigma<\frac{\pi}{\delta}$ and $\nu$ is a symmetric positive Borel probability measure on $[-\sigma/2,\sigma/2]$. It was observed in \cite{Zhang} that
\begin{equation}\label{Framebound}
\cos\Big(\frac{\sigma\delta}{2}\Big)\|f\|_{L^2(\bR)}^2\le\sum_{j\in\bZ}|\tilde{f}(j)|^2\le \|f\|_{L^2(\bR)}^2.
\end{equation}
Thus, $\{\tilde{f}(j):j\in\bZ\}$ is in fact induced by a frame in $\cB_\delta(\bR)$. By the standard frame theory, we are able to completely reconstruct $f$ from the infinite sample data $\{\tilde{f}(j):j\in\bZ\}$ through a dual frame. For studies along this direction, see, for example, \cite{Aldroubi,AST,Aldroubi2005,Grochenig,Sun1, Sun2}.

Motivated by the Shannon sampling theorem, we desire a dual frame that is generated by the shifts of a single function. In other words, we prefer a complete reconstruction formula of the following form
\begin{equation}\label{reconstruction2}
f=\frac1{\sqrt{2\pi}}\sum_{j\in\bZ}\tilde{f}(j)\phi(\cdot-j),\ \ f\in\cB_\delta(\bR),
\end{equation}
where $\phi$ is to be chosen. It has been proved in \cite{Zhang} that if $\phi\in C(\bR)\cap L^2(\bR)$ with $\supp\hat{\phi}\subseteq \bI_\delta:=[-2\pi+\delta,2\pi-\delta]$ then (\ref{reconstruction2}) holds if and only if
\begin{equation}\label{exactreconst}
\hat{\phi}(\xi)W(\xi)=1\mbox{ for almost every } \xi\in[-\delta,\delta],
\end{equation}
where
\begin{equation}\label{Wdefinition}
W(\xi):=\int_{-\sigma/2}^{\sigma/2} e^{it\xi}d\nu(t),\ \ \xi\in [-\delta,\delta].
\end{equation}
Under the above assumptions on $\nu$, $W(\xi)$ is an even function on $[-\delta,\delta]$ and satisfies
\begin{equation}\label{Wrange}
\begin{array}{ll}
&\displaystyle{0<\gamma:=\cos\Big(\frac{\sigma\delta}{2}\Big)\le W(\xi)\le 1,\ |W'(\xi)|\le \frac{\sigma}{2},}\\
&\displaystyle{\hspace{1.7cm} |W''(\xi)|\le \frac{\sigma^2}{4},\ \xi\in[-\delta,\delta].}
\end{array}
\end{equation}

In practice, we only have the finite sample data $\tilde{f}(j)$, $-n<j\le n$. A natural reconstruction method is by
\begin{equation}\label{truncation}
\frac1{\sqrt{2\pi}}\sum_{j=-n+1}^n\tilde{f}(j)\phi(\cdot-j).
\end{equation}
The main purpose of this section is to illustrate by an explicit example that compared to the above direct truncation, regularization by a Gaussian function as follows
\begin{equation}\label{averagesampling}
(\bS_{n,r} f)(t):=\frac1{\sqrt{2\pi}}\sum_{j\in(-n,n]} \tilde{f}(j)\phi(t-j)e^{-\frac{(t-j)^2}{2r^2}}, t\in(0,1), f\in\cB_\delta(\bR)
\end{equation}
may lead to a better convergence rate.

The function $\phi$ satisfying (\ref{exactreconst}) in our example is specified as
\begin{equation}\label{phi}
\hat{\phi}(\xi):=\left\{
\begin{array}{ll}
\frac{1}{W(\xi)}, & \xi\in[-\delta,\delta],\\
\left[\frac{|\xi|-(2\pi-\delta)}{2\pi-2\delta}\right]^2 P(\xi), & \delta\le|\xi|\le 2\pi-\delta,\\
0,&\mbox{elsewhere},
\end{array}
\right.
\end{equation}
where
$$
P(\xi):=\left[\Big(\frac{1}{(\pi-\delta)W(\delta)}-\frac{W'(\delta)}{W^2(\delta)}\Big)|\xi|
+\frac{\pi-2\delta}{(\pi-\delta)W(\delta)}+\frac{\delta W'(\delta)}{W^2(\delta)}\right].
$$
Then $\supp\hat{\phi}\subseteq[-2\pi+\delta,2\pi-\delta]$, $\hat{\phi}\in C^{1}(\bR)$ and $\hat{\phi}'$ is absolutely continuous on $\bR$. By (\ref{Wrange}), we have
\begin{equation}\label{phinorm}
\begin{array}{ll}
&\displaystyle{\|\hat{\phi}\|_{L^\infty(\bI_\delta)}\le \frac{3\pi-\delta}{(\pi-\delta)\gamma}+\frac{\pi\sigma}{\gamma^2},}\\
&\displaystyle{\|\hat{\phi}''\|_{L^\infty(\bI_\delta)}\le \frac{\sigma^2}{\gamma^3}+\frac{8}{(\pi-\delta)^3\gamma}+\frac{4\sigma}{(\pi-\delta)^2\gamma^2},}\\
&\displaystyle{ \|\hat{\phi}''\|_{L^1(\bI_\delta)}\le \frac{2\delta\sigma^2}{\gamma^3}+ \frac{32}{(\pi-\delta)^2\gamma}+\frac{16\sigma}{(\pi-\delta)\gamma^2}.}
\end{array}
\end{equation}
Finally, we shall use an elementary fact from the Fourier analysis \cite{Gasquet} that
\begin{equation}\label{Fourier}
|\phi(t-j)|\le \frac{1}{\sqrt{2\pi}}\frac{\|\hat{\phi}''\|_{L^1(\bI_\delta)}}{|t-j|^2},\ \ j\ge 2,\ t\in(0,1).
\end{equation}

The main result of this section is as follows.

\begin{theorem}\label{Theorem4} Let $\delta\in(0,\pi)$, $n\ge2$, $\gamma=\cos(\frac{\sigma\delta}{2})$, $r=n^{\frac56}$ and $\hat{\phi}$ be given as (\ref{phi}). Then the convergence rate of (\ref{averagesampling}) is
$$
\sup_{f\in\cB_\delta(\bR),\ \|f\|_{L^2(\bR)}\le 1} \|f-\bS_{n,r}f\|_{L^\infty((0,1))} \le c n^{-\frac53},
$$
where
$$
\begin{array}{ll}
&\displaystyle{c=\frac{(\sqrt{\delta}+2\delta)\sigma^2}{\gamma^3}
+\left(\frac{4\sigma\sqrt{\delta}}{(\pi-\delta)^2}+\frac{16\sigma}{\pi-\delta}+\pi\sigma\sqrt{\delta}\right)\frac{1}{\gamma^2}}\\
&\displaystyle{\quad+\left(\frac{8\sqrt{\delta}}{(\pi-\delta)^2}+\frac{32}{\pi-\delta}+10\sqrt{\delta}\right)\frac{1}{(\pi-\delta)\gamma}.}
\end{array}
$$
\end{theorem}
\begin{proof} Let $f\in\cB_\delta(\bR)$ with $\|f\|_{L^2(\bR)}\le 1$. Set $f(t)-\bS_{n,r} f(t):=\bF_1(t)+\bF_2(t)$, $t\in(0,1)$, where
$$
\begin{array}{ll}
&\displaystyle{\bF_1(t):=f(t)-\frac{1}{\sqrt{2\pi}}\sum_{j\in\bZ}\tilde{f}(j)\phi(t-j)e^{-\frac{(t-j)^2}{2r^2}}, }\\
&\displaystyle{\bF_2(t):=\frac{1}{\sqrt{2\pi}}\sum_{j\notin(-n,n]}\tilde{f}(j)\phi(t-j)e^{-\frac{(t-j)^2}{2r^2}}.}
\end{array}
$$

{\it Estimate of $\bF_1$.} By (\ref{reconstruction2}), we have $\hat{f}(\xi)=(\frac{1}{\sqrt{2\pi}}\sum_{j\in\bZ}\tilde{f}(j)\hat{\phi}(\xi)e^{-ij\xi})\chi_{[-\delta,\delta]}(\xi)$, $\xi\in\bR$. Then, for $\xi\in\bR$,
$$
\begin{array}{ll}
\displaystyle{\hat{\bF}_1(\xi)}
&\displaystyle{=\hat{f}(\xi)-\frac{1}{\sqrt{2\pi}} \sum_{j\in\bZ}\tilde{f}(j)\frac{1}{\sqrt{2\pi}}}\\
&\displaystyle{\quad\cdot \int_\bR r e^{-ij(\xi-\eta)}\hat{\phi}(\xi-\eta)e^{-(ij\eta+\frac{r^2\eta^2}{2})}d\eta}\\
&\displaystyle{=\hat{f}(\xi)-\frac{1}{\sqrt{2\pi}}\sum_{j\in\bZ}\tilde{f}(j)e^{-ij\xi} \hat{\phi}(\xi) \frac{1}{\hat{\phi}(\xi)}}\\
&\displaystyle{\quad\cdot\int_\bR \hat{\phi}(\xi-\eta)\frac{1}{\sqrt{2\pi}}r e^{-\frac{r^2\eta^2}{2}}d\eta}\\
&\displaystyle{=\frac{\hat{f}(\xi)\chi_{[-\delta,\delta]}(\xi)}{\hat{\phi}(\xi)}\Big[\hat{\phi}(\xi)-\frac{1}{\sqrt{2\pi}}\int_\bR \hat{\phi}(\xi-\eta)r e^{-\frac{r^2\eta^2}{2}}d\eta\Big].}\\
\end{array}
$$
If $|\xi|> \delta$ then $\hat{\bF}_1(\xi)=0$. If $|\xi|\le \delta$ then by (\ref{lemma2}), (\ref{Wrange}) and (\ref{phi}), we have
$$
\begin{array}{ll}
&\displaystyle{|\hat{\bF}_1(\xi)|}\\
&\displaystyle{\le \frac{|\hat{f}(\xi)|}{\sqrt{2\pi}|\hat{\phi}(\xi)|}\left[ \Big|\int_{|\eta|<\varepsilon} \Big(\hat{\phi}(\xi)-\hat{\phi}(\xi-\eta)\Big)r e^{-\frac{r^2\eta^2}{2}}d\eta\Big|\right.}\\
&\displaystyle{\quad\left.+\Big|\int_{|\eta|\ge\varepsilon} \Big(\hat{\phi}(\xi)-\hat{\phi}(\xi-\eta)\Big)r e^{-\frac{r^2\eta^2}{2}}d\eta\Big|\right]}\\
&\displaystyle{\le \frac{|\hat{f}(\xi)|}{\sqrt{2\pi}}\left[ \Big|\int_{|\eta|<\varepsilon} \Big(\hat{\phi}'(\xi)\eta+\frac{\hat{\phi}''(\xi-\theta \eta)}{2}\eta^2\Big)r e^{-\frac{r^2\eta^2}{2}}d\eta\Big|\right.}\\
&\displaystyle{\quad\left.+\int_{|\eta|\ge\varepsilon} \Big|\hat{\phi}(\xi)-\hat{\phi}(\xi-\eta)\Big| re^{-\frac{r^2\eta^2}{2}}d\eta\right]}\\
&\displaystyle{\le \frac{|\hat{f}(\xi)|}{\sqrt{2\pi}}\left[ \frac{\|\hat{\phi}''\|_{L^\infty(\bI_\delta)}}{2}\cdot\int_{|\eta|<\varepsilon} \eta^2r e^{-\frac{r^2\eta^2}{2}}d\eta\right.}\\
&\displaystyle{\quad\left. +2\|\hat{\phi}\|_{L^\infty(\bI_\delta)} \cdot\int_{|\eta|\ge\varepsilon} r e^{-\frac{r^2\eta^2}{2}}d\eta\right]}\\
&\displaystyle{\le \frac{|\hat{f}(\xi)|}{\sqrt{2\pi}}\left[ \frac{\|\hat{\phi}''\|_{L^\infty(\bI_\delta)}}{2r^2}\cdot\int_{|x|<r\varepsilon}x^2e^{-\frac{x^2}{2}}dx\right.}\\
&\displaystyle{\quad\left.+4\sqrt{2}\|\hat{\phi}\|_{L^\infty(\bI_\delta)}\cdot\int_{\frac{r\varepsilon}{\sqrt{2}}}^{+\infty} e^{-x^2}dx \right]},\\
\end{array}
$$
where $0<\theta<1$ and $\varepsilon:=r$. Combining the last inequality above with $\int_\bR x^2e^{-\frac{x^2}{2}}dx=\sqrt{2\pi}$ and (\ref{lemma1}) yields
$$
\begin{array}{ll}
|\hat{\bF}_1(\xi)|
&\displaystyle{\le \frac{|\hat{f}(\xi)|}{\sqrt{2\pi}}\left(\frac{\sqrt{2\pi}\|\hat{\phi}''\|_{L^\infty(\bI_\delta)}}{2r^2}+\frac{4e^{-\frac{r^4}{2}}\|\hat{\phi}\|_{L^\infty(\bI_\delta)}}{r^2}\right)}\\
&\displaystyle{\le \frac{|\hat{f}(\xi)|}{r^2}\left(\frac{\|\hat{\phi}''\|_{L^\infty(\bI_\delta)}}{2}+\frac{4\|\hat{\phi}\|_{L^\infty(\bI_\delta)}}{\sqrt{2\pi}}\right).}
\end{array}
$$
Bounding $\bF_1$ by the $L^1$-norm of its Fourier transform, we get by the above equation
\begin{equation}\label{E21}
\begin{array}{ll}
|\bF_1(t)|
&\displaystyle{\le\frac{\sqrt{\delta}}{r^2}\left(\frac{\|\hat{\phi}''\|_{L^\infty(\bI_\delta)}}{2\sqrt{\pi}}+\frac{4\|\hat{\phi}\|_{L^\infty(\bI_\delta)}}{\sqrt{2}\pi}\right)}\\
&\displaystyle{\le \frac{\Big(\|\hat{\phi}''\|_{L^\infty(\bI_\delta)}+\|\hat{\phi}\|_{L^\infty(\bI_\delta)}\Big)\sqrt{\delta}}{r^2},\ t\in(0,1).}
\end{array}
\end{equation}

{\it Estimate of $\bF_2$.} By (\ref{Fourier}) and by the Cauchy-Schwartz inequality,  for each $t\in(0,1)$,
$$
\begin{array}{ll}
\displaystyle{|\bF_2(t)|}
&\displaystyle{\le \sum_{j\notin(-n,n]} |\tilde{f}(j)||\phi(t-j)| e^{-\frac{(t-j)^2}{2r^2}}}\\
&\displaystyle{\le\frac{\|\hat{\phi}''\|_{L^1(\bI_\delta)}}{\sqrt{2\pi}n^2}\sum_{j\notin(-n,n]} |\tilde{f}(j)| e^{-\frac{(t-j)^2}{2r^2}}}\\
&\displaystyle{\le \frac{\|\hat{\phi}''\|_{L^1(\bI_\delta)}}{\sqrt{2\pi}n^2}\Big(\sum_{j\notin(-n,n]}|\tilde{f}(j)|^2\Big)^{\frac12}\Big(\sum_{j\notin(-n,n]}e^{-\frac{(t-j)^2}{r^2}} \Big)^{\frac12} }.
\end{array}
$$
It follows from (\ref{lemma2}) and (\ref{Framebound}) that
\begin{equation}\label{E22}
|\bF_2(t)|\le \frac{r\|\hat{\phi}''\|_{L^1(\bI_\delta)}}{\sqrt{2\pi}n^2\sqrt{n-1}}e^{-\frac{(n-1)^2}{2r^2}}\le \frac{r\|\hat{\phi}''\|_{L^1(\bI_\delta)}}{n^{\frac52}} .
\end{equation}
Taking $r=n^{\frac56}$ in (\ref{E21}) and (\ref{E22}), we obtain
$$
\begin{array}{ll}
&\displaystyle{
 \Big\|f-\bS_{n,r}f\Big\|_{L^\infty((0,1))}}\\
&\displaystyle{\le\left(\sqrt{\delta}\|\hat{\phi}''\|_{L^\infty(\bI_\delta)}+
\sqrt{\delta}\|\hat{\phi}\|_{L^\infty(\bI_\delta)}+\|\hat{\phi}''\|_{L^1(\bI_\delta)}\right)\frac{1}{n^{\frac53}},}
\end{array}
$$
which, by (\ref{phinorm}), completes the proof.
\end{proof}

We see that if $\hat{\phi}\in C^{1}(\bR)$ and $\hat{\phi}'$ is absolutely continuous on $\bR$ then the convergence rate of the direcr truncation (\ref{truncation}) is $O(n^{-\frac32})$. Therefore, the Gaussian regularized version (\ref{averagesampling}) indeed results in a better convergence rate.

\section{Numerical Experiments}
We present in this section numerical experiments to illustrate the effectiveness of the Gaussian regularized sampling formula and to demonstrate the validness of our convergence analysis for the formula.

\subsection{Univariate Bandlimited Functions}

Let $\delta<\pi$. The bandlimited function under investigation has the form
\begin{equation}\label{twosincfun}
f_\delta(t)=\frac{1}{\sqrt{\pi(5\delta+4\sin\delta)}}\Big(\frac{2\sin\delta t}{t}+\frac{\sin\delta(t-1)}{t-1}\Big),\ t\in\bR.
\end{equation}
We point out that $f_\delta\in\cB_\delta(\bR)$ and $\|f_\delta\|_{L^2(\bR)}=1$. We shall reconstruct the values of $f$ on $(0,1)$ from $\{f(j):-n+1\le j\le n\}$ by the the truncated Gaussian regularized Shannon sampling series
$$
(S_n f_\delta)(t)=\sum_{j=-n+1}^n f_\delta(j)\sinc(t-j)e^{-\frac{(t-j)^2(\pi-\delta)}{2(n-1)}},\ t\in(0,1).
$$
The error of reconstruction is measured by
$$
E(f_\delta-S_nf_\delta):=\max_{1\le j\le 99}\Big| (f_\delta-S_nf_\delta)\Big(\frac{j}{100}\Big)\Big|
$$
This error is to be compared with the theoretical estimate in Theorem \ref{Theorem1}:
$$
E_{\delta,n}:=\Big(\sqrt{2\delta}+\frac{1}{\sqrt{n}}\Big)\frac{e^{-\frac{(\pi-\delta)(n-1)}{2}}}{\pi\sqrt{(\pi-\delta)(n-1)}}, n\ge2.
$$
We also compute the reconstruction error resulting from directly truncating the Shannon series
$$
E(f_\delta-T_nf_\delta):=\max_{1\le j\le 99}\Big| (f_\delta-T_nf_\delta)\Big(\frac{j}{100}\Big)\Big|.
$$
for comparison, where
$$
(T_n f_\delta)(t):=\sum_{j=-n+1}^n f_\delta(j)\frac{\sin(\pi(t-j))}{\pi(t-j)},\  t\in(0,1).
$$
The above three errors for $n=2,4,\ldots,30$ and $\delta=\frac\pi3,\frac\pi2,\frac{2\pi}3$ are listed in TABLEs \ref{Tab1}, \ref{Tab2}, and \ref{Tab3}, respectively. We also plot $\log E_{\delta,n}$ and $\log E(f_\delta-S_nf_\delta)$ in Fig. \ref{fig1}. We see that the Gaussian regularized Shannon sampling formula converges extremely fast and our theoretical estimate for the convergence rate in Theorem \ref{Theorem1} is pretty sharp.

\begin{table}[htbp]
\centering
\begin{tabular}{|c|c|c|c|} \hline
& $E(f_\delta-T_nf_\delta)$ & $E_{\delta,n}$ &$E(f_\delta-S_nf_\delta)$  \\ \hline
  n=2 & 0.0331 & 0.1663 & 0.0509 \\  \hline
  n=4 & 0.0098 & 0.0107 & 0.0017 \\ \hline
  n=6 & 2.2761e-04 & 9.7124e-04 & 1.0932e-04 \\ \hline
  n=8 & 0.0024 & 9.8103e-05 & 8.5248e-06 \\ \hline
  n=10 & 0.0016 & 1.0433e-05 & 7.2866e-07 \\  \hline
  n=12 & 3.0405e-05 & 1.1440e-06 & 6.8632e-08 \\ \hline
  n=14 & 7.9461e-04 & 1.2799e-07 & 6.6728e-09 \\  \hline
  n=16 & 6.2247e-04 & 1.4525e-08 & 6.5985e-10 \\  \hline
  n=18 & 9.1326e-06 & 1.6661e-09 & 6.8773e-11 \\ \hline
  n=20 & 3.9146e-04 & 1.9267e-10 & 7.2196e-12 \\ \hline
  n=22 & 3.2878e-04 & 2.2428e-11 & 7.5895e-13 \\ \hline
  n=24 & 3.8717e-06 & 2.6246e-12 & 8.3156e-14 \\  \hline
  n=26 & 2.3225e-04 & 3.0851e-13 & 9.2149e-15 \\ \hline
  n=28 & 2.0277e-04 & 3.6399e-14 & 1.4433e-15 \\ \hline
  n=30 & 1.9869e-06 & 4.3080e-15 & 4.4409e-16 \\ \hline
\end{tabular}
\caption{Reconstruction erros for $\delta=\frac{\pi}{3}$.} \label{Tab1}
\end{table}
\begin{table}[htbp]
\centering
\begin{tabular}{|c|c|c|c|} \hline
& $E(f_\delta-T_nf_\delta)$ & $E_{\delta,n}$ &$E(f_\delta-S_nf_\delta)$  \\ \hline
   n=2 &0.0051&0.2871&0.0521  \\\hline
   n=4 &9.6902e-04&0.0316& 0.0030\\\hline
   n=6 &3.2453e-04&0.0049&3.1393e-04 \\ \hline
   n=8 &1.4422e-04&8.3589e-04 &4.2582e-05 \\ \hline
   n=10 &7.5827e-05&1.5055e-04&6.1037e-06  \\\hline
   n=12 &4.4556e-05&2.7936e-05&9.8330e-07 \\\hline
   n=14 &2.8327e-05&5.2865e-06&1.5771e-07 \\ \hline
   n=16 &1.9097e-05&1.0144e-06&2.7379e-08 \\ \hline
   n=18 &1.3472e-05&1.9668e-07&4.6564e-09  \\\hline
   n=20 &9.8525e-06&3.8441e-08&8.4284e-10 \\\hline
   n=22&7.4200e-06&7.5615e-09&1.4863e-10 \\ \hline
   n=24 &5.7257e-06&1.4951e-09&2.7629e-11 \\ \hline
  n=26 & 4.5098e-06 & 2.9691e-10 & 4.9942e-12 \\ \hline
  n=28 & 3.6149e-06 & 5.9176e-11 &9.4547e-13  \\ \hline
  n=30 & 2.9418e-06 & 1.1831e-11 & 1.7397e-13 \\ \hline
\end{tabular}
\caption{Reconstruction errors for $\delta=\frac{\pi}{2}$.} \label{Tab2}
\end{table}
\begin{table}[htbp]
\centering
\begin{tabular}{|c|c|c|c|} \hline
  & $E(f_\delta-T_nf_\delta)$ & $E_{\delta,n}$ &$E(f_\delta-S_nf_\delta)$  \\ \hline
   n=2 &0.0275&0.5074&0.0490  \\\hline
   n=4 &0.0064&0.0951&0.0049 \\\hline
   n=6 &4.6943e-04&0.0249&9.2523e-04 \\ \hline
   n=8 &0.0020&0.0072&1.9165e-04 \\ \hline
   n=10 &0.0012&0.0022&5.0903e-05  \\\hline
   n=12 &6.8750e-05&6.9045e-04&1.3828e-05 \\\hline
   n=14 &6.5098e-04&2.2083e-04&3.5771e-06 \\ \hline
   n=16 &4.6868e-04&7.1605e-05&1.0844e-06 \\ \hline
   n=18 &2.1160e-05&2.3456e-05&3.2408e-07  \\\hline
   n=20 &3.1800e-04&7.7448e-06&9.1110e-08 \\\hline
   n=22 &2.5108e-04&2.5733e-06&2.9088e-08 \\ \hline
   n=24 &9.0566e-06&8.5940e-07&9.0859e-09 \\ \hline
   n=26& 1.8765e-04& 2.8824e-07&2.6710e-09\\ \hline
   n=28& 1.5607e-04&9.7020e-08& 8.7632e-10\\ \hline
   n=30& 4.6691e-06&3.2757e-08&2.8080e-10\\ \hline
\end{tabular}
\caption{Reconstruction errors for $\delta=\frac{2\pi}{3}$.} \label{Tab3}
\end{table}
\vspace{-0.5cm}
\begin{figure}[htbp]
\centering
  \includegraphics[width=4in]{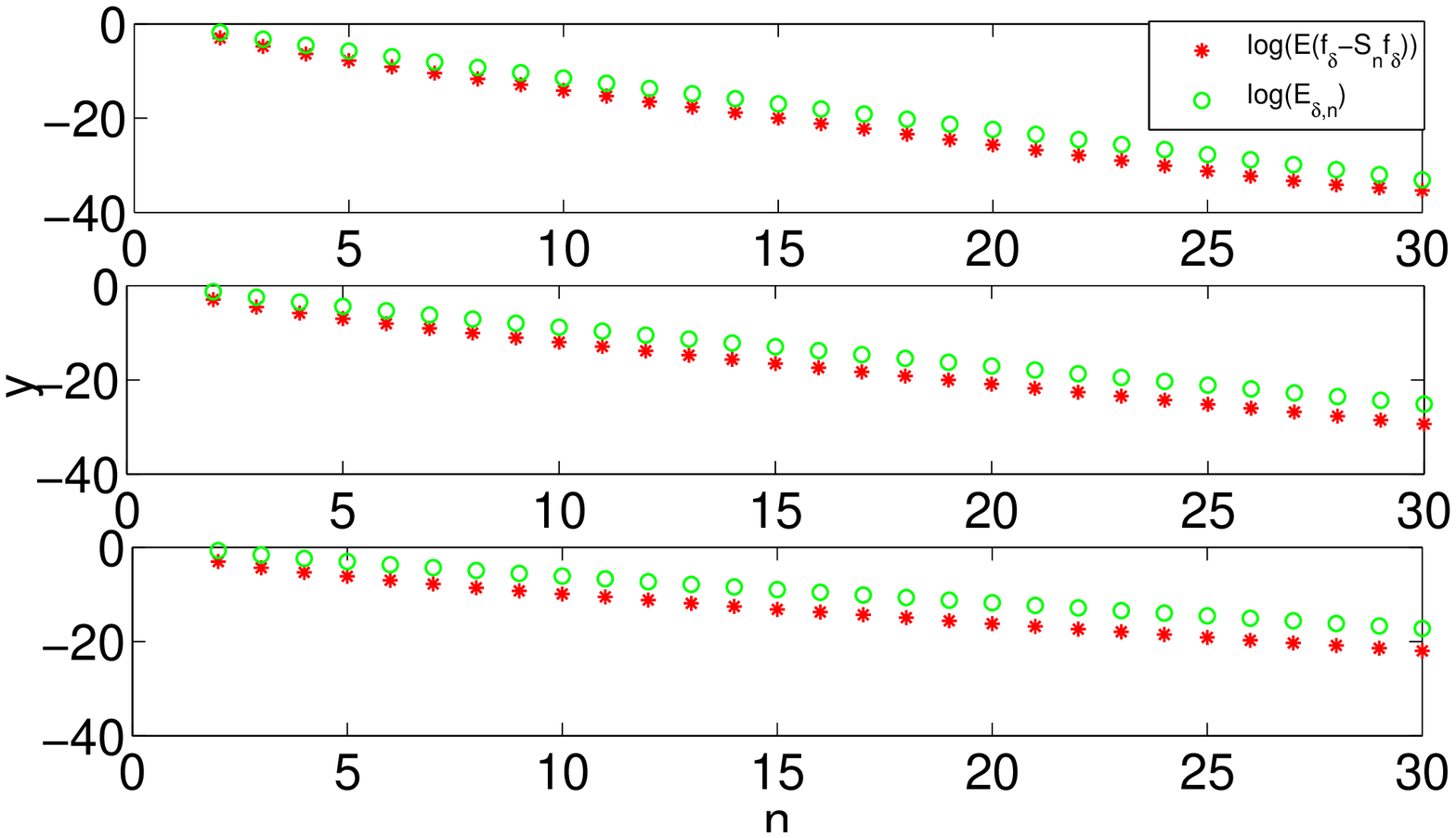}\\
  \caption{Comparison of $\log E_{\delta,n}$ and $\log E(f_\delta-S_nf_\delta)$ for $\delta=\frac{\pi}{3}$ (Top), $\delta=\frac{\pi}{2}$ (Middle), and $\delta=\frac{2\pi}{3}$ (Bottom). Here $n=2,3,4,\dots,30$.}\label{fig1}
\end{figure}

\subsection{Derivative of Univariate Bandlimited Functions}

We now turn to verify the reconstruction error for the first derivative of a univariate bandlimited function. We use $f_\delta$ in the above subsection for simulation. Let
$$
E'_{\delta,n}:=\Big(\sqrt{\frac{2\delta}{3}}\delta+\frac{144}{\sqrt{n}}\Big)\frac{e^{-\frac{(\pi-\delta)(n-2)}{2}}}{\pi\sqrt{(\pi-\delta)(n-2)}}
$$
be the theoretical upper bound established for the derivative in Theorem \ref{Theorem2}. This upper bound together with the errors of the two reconstruction methods are listed in TABLEs \ref{Tab4}, \ref{Tab5}, and \ref{Tab6}. In particular, $\log E'_{\delta,n}$ and $\log E(f'_\delta-(S_nf_\delta)')$ are plotted in Fig. \ref{fig2}.

\begin{table}[htbp]
\centering
\begin{tabular}{|c|c|c|c|} \hline
& $E(f'_\delta-(T_nf_\delta)')$ & $E'_{\delta,n}$ &$E(f'_\delta-(S_nf_\delta)')$  \\ \hline
    n=3 & 0.0064 & 6.4845 & 0.0238 \\ \hline
    n=5 & 0.0192 & 0.3582 & 0.0012 \\ \hline
    n=7 & 0.0103 & 0.0289 & 8.6147e-05 \\ \hline
    n=9 & 2.8936e-04 & 0.0027 & 7.2344e-06 \\ \hline
    n=11 & 0.0041 & 2.6206e-04 & 6.5240e-07 \\ \hline
    n=13 & 0.0030 & 2.6989e-05 & 6.1290e-08 \\ \hline
    n=15 & 6.3904e-05 & 2.8410e-06 & 6.1368e-09 \\ \hline
    n=17 & 0.0017 & 3.0639e-07 & 6.2431e-10 \\ \hline
    n=19 & 0.0014 & 3.3571e-08 & 6.4088e-11 \\ \hline
    n=21 & 2.3438e-05 & 3.7246e-09 & 6.8596e-12 \\ \hline
    n=23 & 9.3372e-04 & 4.1740e-10 & 7.3584e-13 \\ \hline
    n=25 & 8.0083e-04 & 4.7166e-11 & 7.8978e-14 \\ \hline
    n=27 & 1.1058e-05 & 5.3670e-12 & 9.9087e-15 \\ \hline
    n=29 & 5.8806e-04 & 6.1433e-13 & 3.9690e-15  \\ \hline
\end{tabular}
\caption{Reconstruction errors for the first derivative when $\delta=\frac{\pi}{3}$.} \label{Tab4}
\end{table}

\begin{table}[htbp]
\centering
\begin{tabular}{|c|c|c|c|} \hline
& $E(f'_\delta-(T_nf_\delta)')$ & $E'_{\delta,n}$ &$E(f'_\delta-(S_nf_\delta)')$  \\ \hline
    n=3 & 0.0520 & 9.8133 & 0.0321 \\ \hline
    n=5 & 0.0192 & 0.9173 & 0.0028 \\ \hline
    n=7 & 0.0099 & 0.1254 & 3.5204e-04 \\ \hline
    n=9 & 0.0060 & 0.0195 & 4.8058e-05 \\ \hline
    n=11 & 0.0040 & 0.0032 & 7.6107e-06 \\ \hline
    n=13 & 0.0029 & 5.6310e-04 & 1.1889e-06 \\ \hline
    n=15 & 0.0022 & 1.0053e-04 & 2.0547e-07 \\ \hline
    n=17 & 0.0017 & 1.8324e-05 & 3.4305e-08 \\ \hline
    n=19 & 0.0014 & 3.3930e-06 & 6.2129e-09 \\ \hline
    n=21 & 0.0011 & 6.3613e-07 & 1.0797e-09 \\ \hline
    n=23 & 9.2764e-04 & 1.2046e-07 & 2.0128e-10 \\ \hline
    n=25 & 7.8531e-04 & 2.3001e-08 &3.5939e-11  \\ \hline
    n=27 & 6.7338e-04 & 4.4222e-09 & 6.8315e-12 \\ \hline
    n=29 & 5.8378e-04 & 8.5523e-10 & 1.2441e-12 \\ \hline
\end{tabular}
\caption{Reconstruction errors for the first derivative when $\delta=\frac{\pi}{2}$.} \label{Tab5}
\end{table}
\begin{table}[htbp]
\centering
\begin{tabular}{|c|c|c|c|} \hline
& $E(f'_\delta-(T_nf_\delta)')$ & $E'_{\delta,n}$ &$E(f'_\delta-(S_nf_\delta)')$  \\ \hline
    n=3 & 0.0110 & 15.7754 & 0.0419 \\ \hline
    n=5 & 0.0159 & 2.4966 & 0.0071 \\ \hline
    n=7 & 0.0075 & 0.5774 & 0.0014 \\ \hline
    n=9 & 6.3518e-04 & 0.1519 & 3.2892e-04 \\ \hline
    n=11 & 0.0033 & 0.0427 & 9.3779e-05 \\ \hline
    n=13 & 0.0022 & 0.0125 & 2.4603e-05 \\ \hline
    n=15 & 1.4672e-04 & 0.0038 & 6.7163e-06 \\ \hline
    n=17 & 0.0014 & 0.0012 & 2.1156e-06 \\ \hline
    n=19 & 0.0011 & 3.6495e-04 & 6.0578e-07 \\ \hline
    n=21 & 5.4615e-05 & 1.1564e-04 & 1.7733e-07 \\ \hline
    n=23 & 7.5824e-04 & 3.7009e-05 & 5.8201e-08 \\ \hline
    n=25 & 6.1786e-04 & 1.1941e-05 & 1.7429e-08 \\ \hline
     n=27 & 2.5934e-05 & 3.8797e-06 & 5.3009e-09 \\ \hline
    n=29 & 4.7534e-04 & 1.2678e-06 & 1.7783e-09 \\ \hline
\end{tabular}
\caption{Reconstruction errors for the first derivative when $\delta=\frac{2\pi}{3}$.} \label{Tab6}
\end{table}
\begin{figure}[htbp]
\centering
  \includegraphics[width=4.2in]{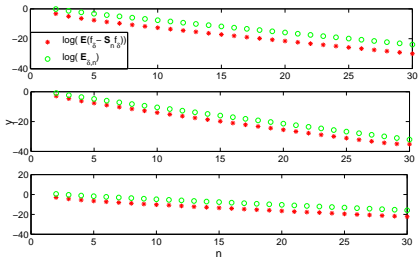}\\
  \caption{Comparison of $\log E'_{\delta,n}$ and $\log E(f'_\delta-(S_nf_\delta)')$
  for $\delta=\frac{\pi}{3}$ (Top), $\delta=\frac{\pi}{2}$ (Middle), $\delta=\frac{2\pi}{3}$ (Bottom).}\label{fig2}
\end{figure}

\subsection{Bivariate Bandlimited Functions}

We consider reconstructing bivariate bandlimited functions by the regularized Shannon sampling formula in this subsection. The function to be reconstructed is given by
$$
f_\delta(t_1,t_2)=f_{\delta_1}(t_1)f_{\delta_2}(t_2),\ t_1,t_2\in\bR,
$$
where $\delta=(\delta_1,\delta_2)$ with $\delta_1\le \delta_2<\pi$, and $f_{\delta_1}$,$f_{\delta_2}$ are defined by (\ref{twosincfun}).  Clearly, $f_\delta\in\cB_\delta(\bR^2)$ and $\|f_\delta\|_{L^2(\bR^2)}=1$. The Gaussian regularized and direct truncation of the Shannon sampling formula are respectively given by
$$
\begin{array}{ll}
({\bf S}_n f_\delta)(t_1,t_2)
&\displaystyle{=\sum_{j_2=-n+1}^n\sum_{j_1=-n+1}^n f_\delta(j_1,j_2)\sinc(t_1-j_1)}\\
&\displaystyle{\quad\cdot\sinc(t_2-j_2)e^{-\frac{((t_1-j_1)^2+(t_2-j_2)^2)(\pi-\delta_2)}{2(n-1)}}}
\end{array}
$$
and
$$
\begin{array}{ll}
({\bf T}_n f_\delta)(t)
&\displaystyle{=\sum_{j_2=-n+1}^n\sum_{j_1=-n+1}^n f_\delta(j_1,j_2)\sinc(t_1-j_1)}\\
&\displaystyle{\quad\cdot\sinc(t_2-j_2).}
\end{array}
$$
Similarly, we shall compare the following two reconstruction errors:
$$
{\bf E}(f_\delta-{\bf S}_nf_\delta):=\max_{1\le j_1,j_2\le 49}\Big| (f_\delta-{\bf S}_nf_\delta)\Big(\frac{j_1}{50},\frac{j_2}{50}\Big)\Big|,
$$
$$
{\bf E}(f_\delta-{\bf T}_nf_\delta):=\max_{1\le j_1,j_2\le 49}\Big| (f_\delta-{\bf T}_nf_\delta)\Big(\frac{j_1}{50},\frac{j_2}{50}\Big)\Big|,
$$
and the theoretical upper bound for the Gaussian regularized method established in Theorem \ref{Theorem3}
$$
{\bf E}_{\delta,n}:=\Big(4\delta_2+\sqrt{\frac2n}\Big)\frac{e^{-\frac{(\pi-\delta_2)(n-1)}{2}}}{\pi\sqrt{(\pi-\delta_2)(n-1)}}, n\ge2.
$$

The results are shown in TABLEs \ref{Tab7}--\ref{Tab9}, and plotted in Fig. \ref{fig3}. Again, we see that our estimate in Theorem \ref{Theorem3} is valid and sharp.

\begin{table*}[htbp]
\centering
\begin{tabular}{|c|c|c|c|}\hline
    & ${\bf E}(f_\delta-{\bf T}_nf_\delta)$ & ${\bf E}_{\delta,n}$ & ${\bf E}(f_\delta-{\bf S}_nf_\delta)$ \\  \hline
  n=2 & 0.0316 & 0.8434 & 0.0378\\ \hline
  n=4 & 4.6243e-04 & 0.0971 & 0.0016 \\ \hline
  n=6 & 0.0041 & 0.0154 & 1.5469e-04 \\ \hline
  n=8 & 1.2792e-04 & 0.0027 & 2.1566e-05 \\ \hline
  n=10 & 0.0014 & 4.8514e-04 & 3.1876e-06 \\ \hline
  n=12 & 2.1342e-05 & 9.0694e-05 & 4.8749e-07 \\  \hline
  n=14 & 7.5353e-04 & 1.7264e-05 & 7.7783e-08 \\  \hline
  n=16 & 1.6774e-05 & 3.3288e-06 & 1.3819e-08 \\  \hline
  n=18 & 4.4541e-04 & 6.4803e-07 & 2.4148e-09 \\  \hline
  n=20 & 4.7213e-06 & 1.2710e-07 & 4.1835e-10 \\  \hline
  n=22 & 3.0412e-04 & 2.5075e-08 & 7.3353e-11 \\  \hline
  n=24 & 5.0184e-06 & 4.9711e-09 & 1.3924e-11 \\ \hline
  n=26 & 2.1432e-04 & 9.8948e-10 & 2.5773e-12 \\  \hline
  n=28 & 1.7325e-06 & 1.9762E-10 & 4.6990e-13 \\  \hline
  n=30 & 1.6324e-04 & 3.9585e-11 & 8.5987e-14 \\  \hline
\end{tabular}
\caption{Comparison of two reconstruction errors and the theoretical upper bound for $\delta=(\frac{\pi}{4},\frac{\pi}{2})$.}\label{Tab7}
\end{table*}

\begin{table*}[htbp]
\centering
\begin{tabular}{|c|c|c|c|}\hline
 & ${\bf E}(f_\delta-{\bf T}_nf_\delta)$ & ${\bf E}_{\delta,n}$ & ${\bf E}(f_\delta-{\bf S}_nf_\delta)$ \\  \hline
  n=2 & 0.0369 & 0.4005 & 0.0558 \\ \hline
  n=4 & 0.0114 & 0.0269 & 0.0020 \\ \hline
  n=6 & 2.6001e-04 & 0.0025 & 1.2529e-04 \\ \hline
  n=8 & 0.0027 & 2.5544e-04 & 9.7707e-06 \\ \hline
  n=10 & 0.0018 & 2.7429e-05 & 8.3512e-07 \\ \hline
  n=12 & 3.4742e-05 & 3.0296e-06 & 7.8652e-08 \\  \hline
  n=14 & 9.1131e-04 & 3.4093e-07 & 7.6473e-09 \\  \hline
  n=16 & 7.1478e-04 & 3.8875e-08 & 7.5614e-10 \\  \hline
  n=18 & 1.0435e-05 & 4.4769e-09 & 7.8805e-11 \\  \hline
  n=20 & 4.4911e-04 & 5.1951e-10 & 8.2729e-12 \\  \hline
  n=22 & 3.7744e-04 & 6.0653e-11 & 8.6942e-13 \\  \hline
  n=24 & 4.4241e-06 & 7.1166e-12 & 9.4758e-14 \\ \hline
  n=26 & 2.6649e-04 & 8.3846e-13 & 1.0381e-14 \\  \hline
  n=28 & 2.3275e-04 & 9.9129e-14 & 1.3323e-15 \\  \hline
  n=30 & 2.2703e-06 & 1.1755e-14 & 5.5511e-16 \\  \hline
\end{tabular}
\caption{Comparison of two reconstruction errors and the theoretical upper bound for $\delta=(\frac{\pi}{3},\frac{\pi}{3})$.}\label{Tab8}
\end{table*}

\begin{table*}[htbp]
\centering
\begin{tabular}{|c|c|c|c|}\hline
   & ${\bf E}(f_\delta-{\bf T}_nf_\delta)$ & ${\bf E}_{\delta,n}$ & ${\bf E}(f_\delta-{\bf S}_nf_\delta)$ \\  \hline
  n=2 & 0.0232 & 1.7279 & 0.0497\\ \hline
  n=4 & 0.0052 & 0.3392 & 0.0040 \\ \hline
  n=6 & 3.0946e-04 & 0.0909 & 6.4739e-04 \\ \hline
  n=8 & 0.0014 & 0.0267 & 1.4029e-04 \\ \hline
  n=10 & 7.9255e-04 & 0.0082 & 3.5117e-05 \\ \hline
  n=12 & 8.1831e-05 & 0.0026 & 9.9855e-06 \\  \hline
  n=14 & 4.7669e-04 & 8.3560e-04 & 2.4699e-06 \\  \hline
  n=16 & 3.4229e-04 & 2.7222e-04 & 7.8074e-07 \\  \hline
  n=18 & 1.3993e-05 & 8.9525e-05 & 2.2438e-07 \\  \hline
  n=20 & 2.1941e-04 & 2.9658e-05 & 6.5528e-08 \\  \hline
  n=22 & 1.7337e-04 & 9.8829e-06 & 2.0155e-08 \\  \hline
  n=24 & 1.0669e-05 & 3.3090e-06 & 6.5146e-09 \\ \hline
  n=26 & 1.3496e-04 & 1.1123e-06 & 1.8527e-09 \\  \hline
  n=28 & 1.1215e-04 & 3.7516e-07 & 6.2749e-10 \\  \hline
  n=30 & 3.0888e-06 & 1.2690e-07 & 1.9506e-10 \\  \hline
\end{tabular}
\caption{Comparison of two reconstruction errors and the theoretical upper bound for $\delta=(\frac{\pi}{2},\frac{2\pi}{3})$.}\label{Tab9}
\end{table*}

\subsection{Average Sampling}

Finally, we briefly discuss about the average sampling. Consider reconstructing the function $f_\delta$ defined by (\ref{twosincfun}) from its finite average sampling data
$$
\widetilde{f_\delta}(j):=\frac{2}{3}f_\delta(j)+\frac{f_\delta(j-\frac18)+f_\delta(j+\frac18)+f_\delta(j-\frac1{16})+f_\delta(j+\frac1{16})}{12}.
$$
By our analysis in Section III, we shall use the following truncated Gaussian regularized reconstruction formula
$$
(\bS_n f_\delta)(t)=\frac{1}{\sqrt{2\pi}}\sum_{j=-n+1}^n \tilde{f}_\delta(j)\phi(t-j)e^{-\frac{(t-j)^2}{2n^{5/3}}},\ t\in (0,1),
$$
where
$$
\begin{array}{ll}
\phi(t)
&\displaystyle{=\sqrt{\frac{2}{\pi}}\left(\int_0^\delta \frac{\cos (t\xi)}{W(\xi)}d\xi+\int_{\delta}^{2\pi-\delta}
\Big(\frac{2\pi-\delta-\xi}{2\pi-2\delta}\Big)^2\right.}\\
&\displaystyle{\quad\cdot\left[\Big(\frac{1}{(\pi-\delta)W(\delta)}-\frac{W'(\delta)}{W^2(\delta)}\Big)\xi+\frac{\pi-2\delta}{(\pi-\delta)W(\delta)}\right.}\\
&\displaystyle{\quad\left.\left.+\frac{\delta W'(\delta)}{W^2(\delta)}\right]\cos(t\xi)d\xi\right).}\\
\end{array}
$$
The reconstruction error is measured by
$$
\bE_{\sigma}(f_\delta-\bS_nf_\delta):=\max_{1\le j\le 19}\Big|(f_\delta-\bS_nf_\delta)\Big(\frac{j}{20}\Big)\Big|.
$$
Results are shown in TABLE \ref{Tab4.31}, which are superior than following estimate in Theorem \ref{Theorem4}
$$
\bE_{\sigma,\delta,n}:=cn^{-\frac53}
$$
and thereby demonstrating validness of the estimate.

\begin{table}[htbp]
\centering
\begin{tabular}{|c|c|c|c|}\hline
      & $\bE_{\sigma,\delta,n}$ & $\bE_{\sigma}(f_\delta-\bS_nf_\delta)$ \\  \hline
n=2&8.9724&0.0175\\  \hline
n=4&2.8261&0.0108\\  \hline
n=6&1.4378&0.0104\\  \hline
n=8&0.8902&0.0100\\  \hline
n=10&0.6137&0.0099\\  \hline
n=12&0.4529&0.0099\\  \hline
n=14&0.3503&0.0098\\  \hline
n=16&0.2804&0.0098\\  \hline
\end{tabular}
\caption{Comparison of reconstruction error and theoretical upper bound when
$\nu=\frac{1}{12}\delta_{-1/8}+\frac{1}{12}\delta_{-1/16}+\frac{2}{3}\delta_0+\frac{1}{12}\delta_{1/16}+\frac{1}{12}\delta_{1/8}$,
$\delta=\frac{\pi}{2}$, and $\sigma=\frac14$.} \label{Tab4.31}
\end{table}

\begin{figure}[htbp]
\centering
  \includegraphics[width=4.2in]{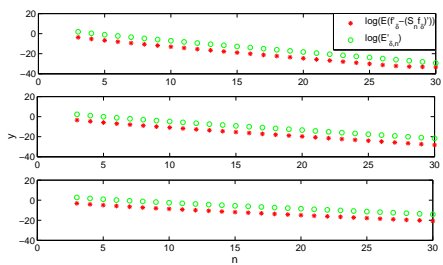}\\
  \caption{Comparison of $\log {\bf E}_{\delta,n}$ and $\log {\bf E}(f_\delta-{\bf S}_nf_\delta)$
  for $\delta=(\frac{\pi}{4},\frac{\pi}{2})$ (Top), $\delta=(\frac{\pi}{3},\frac{\pi}{3})$ (Middle), $\delta=(\frac{\pi}{2},\frac{2\pi}{3})$ (Bottom).}\label{fig3}
\end{figure}

\section{Conclusion}
We present refined convergence analysis for the Gaussian regularized Shannon sampling formula in reconstructing bandlimited functions. The formula is simple, highly efficient, and has received considerable attention in engineering applications. We show in the paper that the formula can achieve by far the best convergence rate among all regularization methods for the Shannon sampling series. Extensive numerical experiments show that our theoretical estimates of the convergence rates are valid and sharp.


\ifCLASSOPTIONcaptionsoff
  \newpage
\fi
\vspace{-0.4in}
\begin{IEEEbiographynophoto}{R. Lin}
received the B.S. degree in mathematics and applied  mathematics from Zhangzhou Normal University, Zhangzhou, P.R. China, in 2012.  He is currently pursuing the Ph.D. degree in computational mathematics at Sun Yat-sen University, Guangzhou, P.R. China, through a five-year Ph.D. program.

His research interests include time-frequency analysis and kernel methods in machine learning.
\end{IEEEbiographynophoto}
\vspace{-0.4in}
\begin{IEEEbiographynophoto}{H. Zhang}
received the B.S. degree in mathematics and applied  mathematics from Beijing Normal University, Beijing, P.R. China, in 2003, the M.S. degree in computational mathematics from the Chinese Academy of Sciences, Beijing, P. R. China in 2006, and the Ph.D. degree in mathematics from Syracuse University, NY, in 2009.

From June 2009 to May 2010, he was a postdoctoral research fellow at University of Michigan, Ann Arbor. Since June 2010, he has been a Professor with the School of Mathematics and Computational Science, Sun Yat-sen University, Guangzhou, P. R. China.

Prof. Zhang's research interests include computational harmonic analysis, kernel methods in machine learning, and sampling theory.
\end{IEEEbiographynophoto}









\end{document}